\documentclass[12pt]{article}

\usepackage{amssymb,fullpage}
\usepackage{amsmath}

\usepackage[T1]{fontenc}
\usepackage[utf8]{inputenc}
\usepackage{amsfonts,amsthm,mathtools,mathrsfs}
\usepackage{braket,physics}
\usepackage{comment}
\usepackage{indentfirst}
\usepackage{enumerate}
\usepackage{booktabs}
\usepackage{tabularx}
\usepackage{makecell}
\usepackage{stmaryrd}
\usepackage{here} 
\usepackage[numbers,sort&compress]{natbib}
\usepackage{graphicx}
\usepackage{hyperref}

\newcommand{\bP}{\mathbb P}
\newcommand{\Dpa}{\mathrm{D_{pa}}}
\newcommand{\Dop}{\mathrm{D_{op}}}
\newcommand{\Cpa}{\mathrm{C_{pa}}}
\newcommand{\Cop}{\mathrm{C_{op}}}
\newcommand{\sep}{; }
\DeclareMathOperator*{\bLor}{\bigvee}

\newtheorem{thm}{Theorem}[section]

\title{Modeling the Disjunction Effect within Classical Probability:\\ A New Decision Process Model and Comparison with Quantum-like Models}
\author{%
  Ryo Nasu\thanks{Corresponding author. Email: \href{mailto:nasu.ryo.x0@s.mail.nagoya-u.ac.jp}{nasu.ryo.x0@s.mail.nagoya-u.ac.jp}}\\[0.4ex]
  \small Graduate School of Informatics\\
  \small Nagoya University\\
  \small Nagoya, Japan
  \and
  Yoshihiro Maruyama\\[0.4ex]
  \small Department of Mathematical and Information Sciences\\
  \small Kyoto University\\
  \small Kyoto, Japan
}
\date{}

\begin{document}

\maketitle

\begin{abstract}
In a Prisoner's Dilemma experiment, participants defected less often when the opponent's move was unknown than when they were told that the opponent had defected or that the opponent had cooperated.
This disjunction effect has often been interpreted as a violation of the classical law of total probability and has motivated quantum-like models of decision making.
We re-examine this interpretation and identify a strong assumption implicit in the conventional classical model.
The conventional model uses the two informed conditions to represent the full range of participants' expectations and thus treats every participant as certain of the opponent's move even when that move is unknown.
We relax this premise by introducing a continuous parameter for each participant's degree of expectation that the opponent will defect.
Participants are grouped by expectation level, with the intermediate groups representing ambiguous expectations.
A standard quantum-like model offers a different account of this same uncertainty about the opponent's move.
In that model, ambiguity arises through superposition, and ambiguous pure states are generic (dense and of full unitarily invariant measure on the unit sphere), so certainty states are mathematically exceptional.
Despite this difference in how ambiguity is represented, the proposed classical model can reproduce the observed strong disjunction effect while satisfying the classical law of total probability.
More generally, we prove that it can exactly reproduce any triple of defection rates strictly between zero and one, including every such triple generated by the quantum-like model.
The proposed model is therefore no less expressive than the quantum-like model for these three observable rates while remaining entirely within classical probability, and the substantive difference between the two lies in how they represent ambiguity and interpret events.
\end{abstract}



\section{Introduction} \label{section:introduction}
The disjunction effect, introduced by Tversky and Shafir \cite{Shafir_1992_1}, refers to a violation of the sure-thing principle, one of the fundamental axioms of rational decision-making. 
According to Savage's sure-thing principle, if a person prefers a particular option both when an event occurs and when it does not occur, then the person should also prefer the same option when it is unknown whether the event occurs or not \cite{Savage_1954}.  
In their experiment \cite{Shafir_1992_1}, participants were asked to imagine that they had just played a gamble in which they had a 50\% chance of winning \$200 and a 50\% chance of losing \$100. 
Participants were then asked whether they would accept a second, identical gamble. 
Depending on the scenario, they were told that they had won the first gamble, that they had lost it, or that they would learn its outcome only after making their choice.
The question was whether knowing the outcome of the first gamble affected participants' willingness to take the second. 
Most participants accepted the second gamble after a reported win or a reported loss. 
When the outcome was unknown, however, most declined.
The authors interpreted this behavior as a violation of the sure-thing principle, and called it the disjunction effect \cite{Shafir_1992_1}.

In a subsequent study, Shafir and Tversky \cite{shafir_1992_2} demonstrated that the disjunction effect can also occur in a Prisoner's Dilemma (PD) game. 
Participants played a series of one-shot games and were told that they faced a different student in each game.
Each player could choose to either defect or cooperate. 
If both players cooperated, they each received a moderate payoff. 
If both defected, they each received a lower payoff. 
If one player defected while the other cooperated, the defector received the highest payoff, and the cooperating player received the lowest. 
Each participant played several PD games, each presented under three different information conditions.
In one condition, participants were told that the opponent had defected, and they chose to defect on 97\% of trials.
In another, they were told that the opponent had cooperated, and they chose to defect on 84\% of trials.
In the third condition, participants were not informed of the opponent's choice, and the defection rate was only 63\%.
This result was also regarded as an instance of the disjunction effect \cite{shafir_1992_2}.

These findings, originally presented as evidence of violations of Savage's sure-thing principle, have sparked considerable discussion.
In particular, subsequent studies have interpreted the observed pattern of choice rates as a challenge to the classical law of total probability (LTP).
In the PD experiment, this interpretation treats the defection rates observed in the two informed conditions as the conditional probabilities of defection when the opponent's choice is unknown.
Under that assumption, the LTP gives a weighted average of 84\% and 97\%, whereas the observed rate was 63\%.
Quantum-like accounts use interference to explain how a choice probability can fall outside the range of the two informed rates.
Khrennikov \cite{Khrennikov_2006} studied cognitive interference using a modified total-probability formula and a Hilbert-space representation of mental states.
Busemeyer et al.~\cite{Busemeyer_2006} and Pothos and Busemeyer~\cite{Pothos_2009} developed Hilbert-space models of the disjunction effect as alternatives to classical models.\footnote{For broader discussions of quantum cognition, see \cite{BB12,Kh23,M1,M2}.}

Several authors have also analyzed the disjunction effect within classical probability theory.
Focusing on the PD experiment, Xin et al.  \cite{Xin_2022} and Moreira and Wichert \cite{Moreira_Wichert_2018} show that the observed pattern can be reproduced by suitably enriching classical models—through more realistic Markov dynamics or classical Bayesian networks with latent variables—without abandoning classical Kolmogorovian probability. 
In another direction, Gelastopoulos and Le Mens \cite{Gelastopoulos_2024} argue that the usual way of reading the disjunction effect as a violation of the classical LTP is itself problematic. 
These works indicate that the apparent tension between the disjunction effect and classical probability may have as much to do with how the decision process and events are modeled as with the probabilistic framework per se. 
The present paper also belongs to this general line of classical re-examination, but it does so by rethinking how the opponent's choice is interpreted and represented in the PD disjunction-effect analysis. 

The starting point of the present paper is that, in modeling the PD disjunction-effect experiments, an important distinction between two readings of the events describing the opponent's choice has often not been made explicit.
Under an objective (\emph{world-event}) reading, these events describe the opponent's actual action in the task; under a subjective (\emph{mental-event}) reading, they describe the participant's internal expectation about the opponent's move. 
Our analysis then focuses on the mental-event reading, which is well suited to modeling the decision process itself. 
Under this reading, the defection rate in the unknown condition is identified with the unconditional probability of defection, and the two informed rates with the probabilities of defection conditional on the corresponding expectations.
We interpret this use of the informed rates as treating an expectation of a particular move as certainty about that move.
Given this interpretation, the standard partition assumptions on the opponent-choice events—that they are mutually exclusive and exhaustive—entail that, prior to choosing, each participant is already in one of two certainty states: certain the opponent will defect or certain the opponent will cooperate. 
This implicit \emph{certainty-only} premise in the conventional decision-process model forces an apparent conflict with the PD data. 
The strength and consequences of this premise do not appear to have been made explicit in the existing literature; they will be spelled out in Section \ref{sec:background}.

The present paper makes this implicit modeling choice explicit and replaces it with a more flexible description of mental events. 
We introduce a new classical (non-dynamical) decision-process model (DPM) in which each participant is endowed with a continuous expectation parameter $\bP\in[0,1]$ representing the anticipated likelihood that the opponent will defect, and the participant pool is partitioned according to levels of $\bP$.
With this simple partition, the model yields a classical LTP decomposition \eqref{eq:LTP_P} for the overall defection rate that is compatible with the PD findings. 
Any empirically observed triple of defection rates in $(0,1)^3$, including triples that exhibit a strong disjunction effect, can be realized by an instance of this model that satisfies the classical LTP (Theorem~\ref{thm:disjunction_realization}).

In parallel, a standard quantum-like formulation of the PD disjunction effect is examined from the same decision-process perspective. 
Formulated in a common notation, the quantum-like model produces a quantum LTP with an interference term and the same three observable defection rates. 
We prove (Theorem~\ref{cor:quantum_disjunction_realization}) that for any such triple generated by the quantum-like model there exists an instance of the classical DPM that reproduces it exactly. 
Comparing the two models side by side highlights that both classical and quantum-like frameworks allow decisions under ambiguous expectations, but realize this capacity in systematically different ways that reflect the familiar contrast between classical (set-based) and quantum (Hilbert-space) interpretations of events and probabilities.

\paragraph{Contributions}
The main contributions of this paper are as follows:
\begin{enumerate}
  \item It is shown that, under the mental-event reading of the
        opponent-choice events and the usual identification of information
        conditions with conditional probabilities, the conventional classical
        DPM implicitly assumes a strong certainty-only premise via the
        standard partition assumptions, and that this hidden assumption, rather
        than the classical LTP itself, can be viewed as the source of the
        apparent conflict with the PD disjunction-effect data.

  \item A non-dynamical classical DPM with a continuous expectation parameter
        $\bP$ is proposed that preserves the LTP decomposition
        \eqref{eq:LTP_P}, remains consistent with the PD findings, and
        contains the conventional DPM as a boundary case. The main existence
        result (Theorem~\ref{thm:disjunction_realization}) shows that for each
        triple of defection rates in $(0,1)^3$, there is a three-bin instance
        of this model that realizes it while satisfying the classical LTP. In contrast to existing classical PD models based on
        enriched dynamics or latent processes \cite{Xin_2022, Moreira_Wichert_2018},
        this is achieved within a classical event structure by refining
        the partition of mental events, without modeling explicit temporal
        evolution.

  \item The semantic role of ambiguity is made explicit and compared across
        the three DPMs at the level of event semantics. We define the
        ambiguity set
        \(
        A := \llbracket \top\rrbracket \setminus
        \bigl(\llbracket D_{\mathrm{op}}\rrbracket \cup \llbracket C_{\mathrm{op}}\rrbracket\bigr),
        \)
        and Theorem~\ref{thm:ambiguity-sets} shows that (a) in the conventional classical
        DPM one has $A=\varnothing$ (ambiguity is excluded by construction);
        (b) in the proposed classical DPM, $A$ is exactly the union of the interior expectation bins;
        and (c) in the quantum-like DPM, the ambiguous pure states form a dense set of
        full unitarily invariant measure on the unit sphere (certainty states are exceptional).

  \item The proposed classical DPM and a standard quantum-like DPM for the PD
        disjunction effect are formulated in a common notation and compared
        semantically. Theorem~\ref{cor:quantum_disjunction_realization}
        establishes that, at the level of the three observable defection
        rates, any pattern produced by the quantum-like model can be matched
        by an instance of the classical model. This highlights their
        shared ability to accommodate ambiguous expectations and clarifies
        that their differences arise from their respective interpretations of
        events and probabilities and from where they locate ambiguity
        (heterogeneity across participants versus superposition within a
        cognitive state), rather than from raw probabilistic expressiveness.
\end{enumerate}
Note that, to the best of our knowledge, mathematical results of the form
Theorem~\ref{thm:disjunction_realization}, Theorem~\ref{cor:quantum_disjunction_realization},
or Theorem~\ref{thm:ambiguity-sets} have not been known in the quantum cognition literature.

\paragraph{Organization}
Section~\ref{sec:background} reviews the standard LTP-based formulation of the
PD disjunction effect, introduces the two readings of the opponent-choice
events, and makes the conventional classical DPM under the mental-event
reading explicit, including its implicit certainty-only premise.
Section~\ref{section:classic} then develops the revised classical DPM by
relaxing that premise, introducing the expectation parameter $\bP$, and
deriving the corresponding LTP decomposition~\eqref{eq:LTP_P}.
Section~\ref{section:quantum} reviews the quantum-like formalism, states the
quantum analogue of the LTP, formulates the associated quantum-like DPM in a
common notation, and uses it for a side-by-side comparison with the proposed
classical model. Section~\ref{mathresults} presents the supporting
mathematical results, including the existence theorem for classical
realizations of disjunction-effect triples and its consequence for quantum
disjunction probabilities. In Section~\ref{sec:semantics_comparison}, we complement Section~\ref{section:quantum}
by contrasting the event semantics underlying the conventional classical 
DPM, the proposed classical DPM, and the quantum-like DPM, thus proving the aforementioned Theorem~\ref{thm:ambiguity-sets}.

\section{Background and setup for the proposed decision-process model}
\label{sec:background}

\subsection{Review of arguments for violations of the classical law of total probability} 
\label{subsec:review_LTP_violation}

In this section, we review key arguments that support the view that the disjunction effect constitutes a violation of the classical LTP (see, e.g., \cite{Pothos_2022} for a review). Throughout this paper, we focus on the PD experiment; nevertheless, we expect that the basic diagnosis developed here extends more broadly to other disjunction-effect paradigms, even if the details may need to be reformulated from case to case.

We denote by $\Dpa$ the event in which the participant chooses to defect. $\Dop$ and $\Cop$ denote the events in which the virtual opponent chooses to defect and cooperate, respectively; the precise interpretation of those events will be discussed in Section \ref{subsec:events_interpretation}. We write $\Pr(E\mid F)$ for the probability of event $E$ conditional on event $F$. The classical LTP describing the participant's choice is given by
\begin{equation}\label{eq:LTP_DC}
    \Pr (\Dpa) = \Pr (\Dpa \mid \Dop) \cdot \Pr(\Dop) + \Pr (\Dpa \mid \Cop) \cdot \Pr(\Cop) ,
\end{equation}
under the assumptions that 
\begin{equation}\label{eq:empty_DC}
    \Dop \land \Cop \equiv \varnothing,
\end{equation}
\begin{equation}\label{eq:sure_DC}
    \Dop \lor \Cop \equiv \top
\end{equation}
hold. Here, $\varnothing$ and $\top$ denote the empty event and the sure event, respectively, and $\equiv$ denotes logical equivalence. Assumption \eqref{eq:empty_DC} states that the two events $\Dop$ and $\Cop$ are incompatible, in the sense that they cannot both hold simultaneously. Assumption \eqref{eq:sure_DC} states that the disjunction of $\Dop$ and $\Cop$ coincides with the sure event, meaning that at least one of them must hold in every possible situation. Together, $\Dop$ and $\Cop$ form a complete partition of the event space: one of them must occur. In particular, these assumptions imply that
\begin{equation}\label{eq:sum1_DC}
\Pr(\Dop) + \Pr(\Cop) = 1.
\end{equation}
Thus, equation \eqref{eq:LTP_DC} is taken to imply that $\Pr(\Dpa)$ should lie between $\Pr(\Dpa \mid \Dop)$ and $\Pr(\Dpa \mid \Cop)$.

To connect these equations to the experimental data, the literature typically identifies the observed defection rates in the three distinct conditions—``unknown'', ``told defect'', and ``told cooperate''—with the following theoretical quantities, respectively:
\begin{equation} \label{eq:identification}
\Pr(\Dpa)\ \text{(``unknown'')},\quad
\Pr(\Dpa\mid \Dop)\ \text{(``told defect'')},\quad
\Pr(\Dpa\mid \Cop)\ \text{(``told cooperate'')}.
\end{equation}

Under this identification, however, the experimental results show that $\Pr(\Dpa)$ is smaller than both $\Pr(\Dpa \mid \Dop)$ and $\Pr (\Dpa \mid \Cop)$, as noted in Section \ref{section:introduction}, which contradicts equation \eqref{eq:LTP_DC}. 
This has been cited as evidence that the classical LTP is violated in the experiment.

\subsection{Two readings of the opponent-choice events}
\label{subsec:events_interpretation}

This line of reasoning has been repeated by many authors, as a claim that the classical LTP is violated in the disjunction-effect experiment. However, the precise interpretation of the opponent-choice events—namely $\Dop$ and $\Cop$—is often left implicit. A closer examination suggests that the literature divides into two readings. In this section, we clarify what these events are intended to denote; this examination is specific to this paper and forms the first step toward the new DPM developed in Section \ref{section:classic}. 

The interpretation of the events $\Dop$ and $\Cop$ divides into the following two readings:
\begin{description}
    \item[(i) \emph{Objective} (\emph{world-event}) reading.] $\Dop$ and $\Cop$ are taken to denote objective events concerning the virtual opponent's actual action in the (hypothetical) world of the task: the opponent defects or cooperates (e.g., \cite{Moreira_Wichert_2018}).
    \item[(ii) \emph{Subjective} (\emph{mental-event}) reading.] $\Dop$ and $\Cop$ are treated as internal events in which the participant expects, prior to making a decision, that the opponent will defect or cooperate (e.g., \cite{Khrennikov_2009, Busemeyer_2006, Tesar_2020}).
\end{description}

With respect to the objective reading (i), however, Gelastopoulos and Le Mens \cite{Gelastopoulos_2024} argued that equation \eqref{eq:LTP_DC} does not contradict the experimental findings. Although their paper is framed as a general rebuttal to claims of the classical LTP violation, their core argument seems to presuppose the interpretation (i). Also, Khrennikov \cite{Khrennikov_2009} had earlier stated that the apparent violation of the classical LTP is not surprising, from a closely related viewpoint.
The key point in Gelastopoulos and Le Mens \cite{Gelastopoulos_2024} is that the experimental manipulation concerns information rather than world events: conditioning on $\Dop$ or $\Cop$ concerns whether the (virtual) opponent actually defected or cooperated, whereas the ``told defect/cooperate'' conditions concern what the participant is told. Hence the identification in \eqref{eq:identification} is not appropriate under reading (i), as it conflates conditioning on world events with conditioning on informational conditions.
When the opponent's choice is not disclosed, it is standard to assume that $\Dpa$ is independent of $\Dop$ and $\Cop$, which yields
\begin{equation}
\Pr(\Dpa \mid \Dop) = \Pr(\Dpa), \qquad
\Pr(\Dpa \mid \Cop) = \Pr(\Dpa).
\end{equation}
Moreover, since the opponent's choice is either to defect or to cooperate, adopting assumptions \eqref{eq:empty_DC} and \eqref{eq:sure_DC} is natural and reasonable; in particular, \eqref{eq:sum1_DC} follows. Consequently, equation \eqref{eq:LTP_DC} is necessarily satisfied regardless of the experimental results.

In this way, equation \eqref{eq:LTP_DC} does not contradict the experimental results when $\Dop$ and $\Cop$ are interpreted as objective events concerning the opponent's action; however, under this interpretation, equation \eqref{eq:LTP_DC} is reduced to a trivial identity and is not helpful as a mathematically informative constraint on the decision process associated with the experiment. Accordingly, we do not pursue interpretation (i) further and focus on interpretation (ii) in what follows.

\subsection{The conventional decision-process model} \label{subsec:decision-process_model_review}

In this section, we consider reading (ii), under which $\Dop$ and $\Cop$ are interpreted as the participant's mental events. Recall that, on this reading, the participant forms an expectation about the opponent's action prior to choosing their own action: expecting defection corresponds to $\Dop$, and expecting cooperation corresponds to $\Cop$ \cite{Tesar_2020}.
When assumptions \eqref{eq:empty_DC} and \eqref{eq:sure_DC} are adopted and \eqref{eq:LTP_DC} is used to model the participant's choice under this reading, this amounts to adopting a particular premise about the decision process. Our aim here is to make that implicit DPM explicit, clarify its assumptions, and set up the new DPM developed in Section \ref{section:classic}.

Assuming the identification in \eqref{eq:identification}, the intended interpretation of $\Dop$ and $\Cop$ becomes more specific. 
In assigning defection probabilities, the model therefore treats a participant who expects defection as if they knew that the opponent had defected. The same applies to cooperation. In other words, expectation is treated as if it were knowledge. 
In this reading, ``expectation'' is not a mild inclination about what is likely, but a \emph{certainty} that the corresponding outcome will occur.

With this interpretation in place, the implications of \eqref{eq:empty_DC} and \eqref{eq:sure_DC} can be stated clearly. Assumption \eqref{eq:empty_DC} entails that a participant cannot simultaneously hold a certainty that the opponent will defect and a certainty that the opponent will cooperate. Given this mutual exclusivity, assumption~\eqref{eq:sure_DC} entails that, before choosing, each participant is certain of exactly one of these possibilities---defection or cooperation. In particular, it rules out participants who proceed with an \emph{ambiguous} or \emph{mixed expectation} about the opponent's action.

For concreteness, Figure~\ref{fig:DPM} (a) visualizes the conventional DPM under reading (ii) together with the identification in \eqref{eq:identification}. The top panel shows the participant pool. The middle layer splits it into two subpopulations—those certain the opponent will defect and those certain the opponent will cooperate—reflecting the implications of \eqref{eq:empty_DC} and \eqref{eq:sure_DC}. The bottom layer maps each subpopulation into the two actions—defection via $\Pr(\Dpa\mid \Dop)$ and $\Pr(\Dpa\mid \Cop)$, and cooperation via $\Pr(\Cpa\mid \Dop)$ and $\Pr(\Cpa\mid \Cop)$ (with $\Cpa$ denoting cooperation). Under these assumptions, the pipeline instantiates the classical LTP in \eqref{eq:LTP_DC}.

Our focus in Section \ref{section:classic} will be on assumptions \eqref{eq:empty_DC} and \eqref{eq:sure_DC}. The basic idea is to replace the requirement that every participant hold a certain expectation about the opponent's action with the more realistic premise that participants may entertain ambiguous expectations that allow for both possibilities to varying degrees. Building on this premise, Section \ref{section:classic} develops a new DPM that remains within the classical framework yet does not conflict with the observed disjunction-effect pattern.

\begin{figure}[p]
  \centering

  \includegraphics[width=0.74\textwidth]{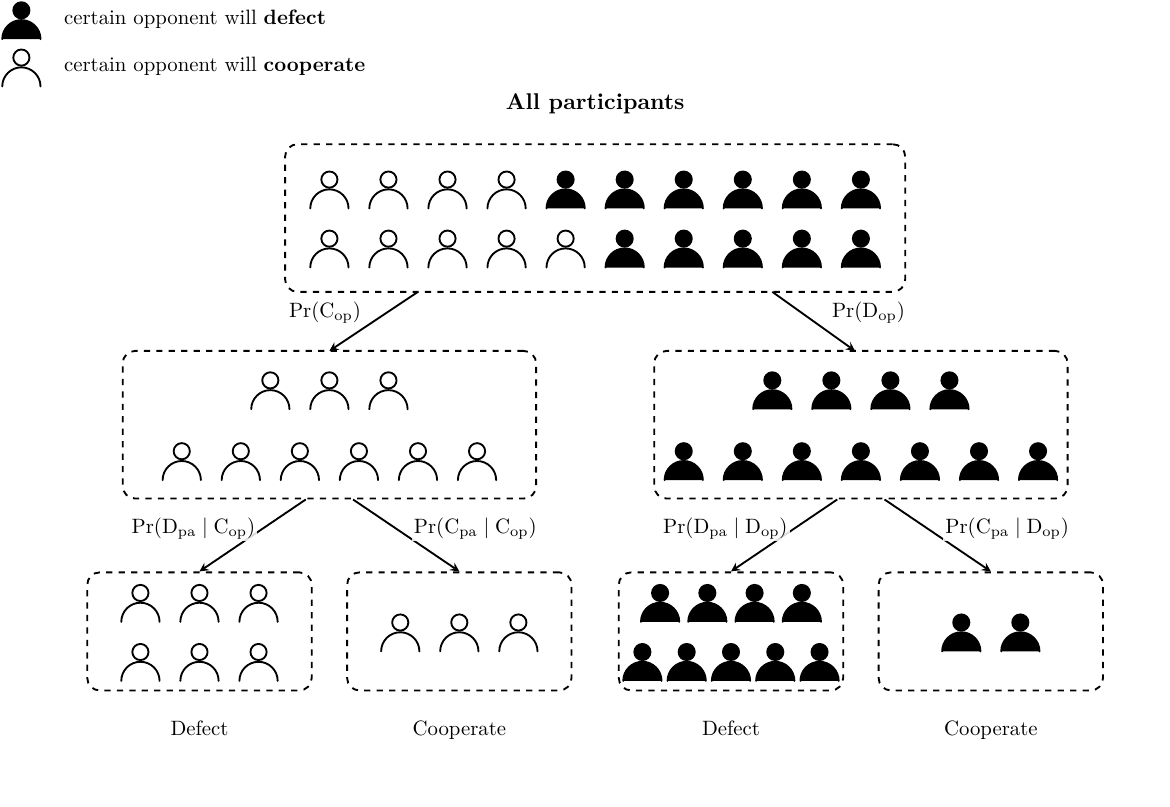}\\[1mm]
  \textbf{(a)} Conventional decision-process model  

  \vspace{2mm}

  \includegraphics[width=0.74\textwidth]{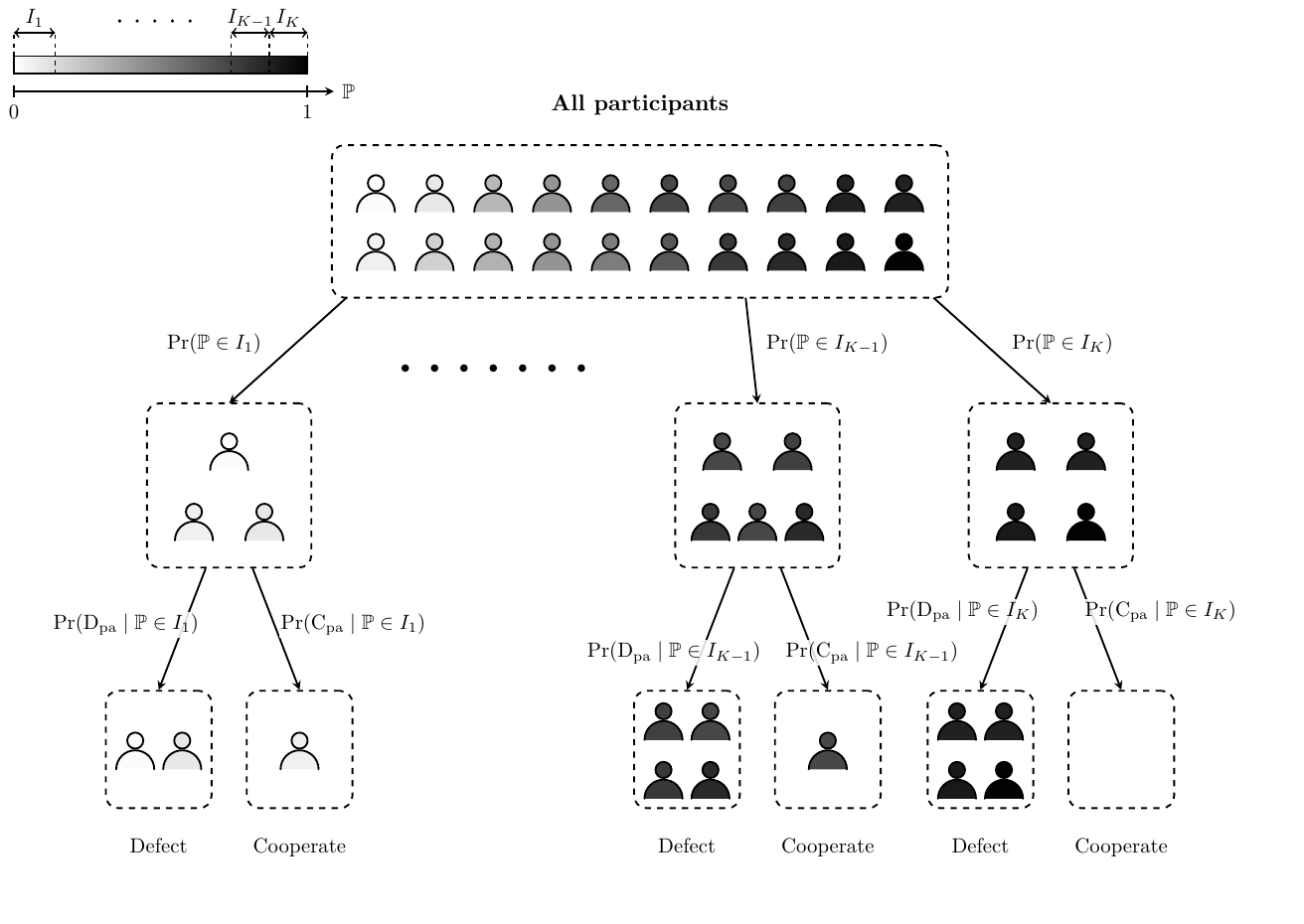}\\[1mm]
  \textbf{(b)} Proposed decision-process model

   \caption{Classical decision-process models.
   (a) Conventional model. Participants are split into two certainty states, ``opponent defects'' ($\Pr(\Dop)$) or ``opponent cooperates'' ($\Pr(\Cop)$), as implied by \eqref{eq:empty_DC}, \eqref{eq:sure_DC}; choices then satisfy equation \eqref{eq:LTP_DC}.
   (b) Proposed model. Participants are grouped by an expectation parameter $\bP\in[0,1]$ into bins $I_k$ with weights $\Pr(\bP\in I_k)$; actions follow $\Pr(\Dpa\mid\bP\in I_k)$, yielding \eqref{eq:LTP_P} under \eqref{eq:empty_P}, \eqref{eq:sure_P}. See Sections \ref{subsec:decision-process_model_review} and \ref{section:classic} for details.}
   \label{fig:DPM}
\end{figure}

\paragraph{Related work: classical accounts of the disjunction effect}
The present approach is related to, but distinct from, recent attempts to account for the PD disjunction effect within a classical framework. 

Xin et al. \cite{Xin_2022} propose a more realistic Markov belief-action model for the PD experiment, explicitly targeting unrealistic assumptions about the decision process and emphasizing subjective uncertainty via a degree of subjective uncertainty (DSN) based on Shannon entropy. They retain the standard four belief-action events and the usual partition assumptions
(corresponding to \eqref{eq:empty_DC} and \eqref{eq:sure_DC}), and incorporate uncertainty through DSN in the dynamics of the Markov process. 

Moreira and Wichert \cite{Moreira_Wichert_2018} investigate whether introducing latent type variables into a classical Bayesian network suffices to reproduce the disjunction-effect probabilities observed in the PD data; they show that classical Bayesian networks with latent variables (within a Kolmogorovian framework) can be constructed so as to match this pattern, but only at the cost of increasingly complex latent structure and without a single simple model covering both known and unknown conditions, which motivates their move to a quantum-like network. Their analysis is conducted under the world-event reading and keeps \eqref{eq:empty_DC} and \eqref{eq:sure_DC} in place. 

By contrast, the present paper remains non-dynamical and does not introduce additional latent variables; it stays within a classical probability framework but changes the modeling of mental events themselves. By relaxing the certainty-only partition on opponent-choice expectations, rather than enriching the dynamics or adding hidden structure, it shows that the PD disjunction-effect pattern can be reconciled with the classical LTP in a comparatively simple way.

\section{Proposal of a new classical decision-process model} \label{section:classic}
We now introduce a DPM that remains fully classical yet accommodates the disjunction-effect pattern. In Section \ref{subsec:decision-process_model_review} and Figure~\ref{fig:DPM} (a), we explained that the conventional DPM implicitly assumes that every participant forms prior certainty about the opponent's choice before deciding their own action—either the opponent will definitely defect or will definitely cooperate. In contrast, we allow participants to hold ambiguous expectations about the opponent's action.

We introduce a continuous parameter $\bP \in [0,1]$ to represent a participant's degree of expectation that the opponent will defect. Each participant has their own value of $\bP$; values closer to $1$ indicate stronger belief in defection, and values closer to $0$ indicate stronger belief in cooperation.

Formally, choose $0=a_0<a_1<\cdots<a_K=1$ and set $I_k=[a_{k-1},a_k)$ for $1\le k<K$ and $I_K=[a_{K-1},1]$. Our basic assumption is that the events $\{\,\bP\in I_k\,\}_{k=1}^K$ form a partition, i.e.,
\begin{equation}\label{eq:empty_P}
\forall\,k\neq \ell:\quad [\bP\in I_k]\;\land\;[\bP\in I_\ell] \equiv \varnothing,
\end{equation}
\begin{equation}\label{eq:sure_P}
\bLor_{k=1}^{K}\,[\bP\in I_k] \equiv \top.
\end{equation}

These are the counterparts of assumptions \eqref{eq:empty_DC} and \eqref{eq:sure_DC} in the conventional DPM. Under \eqref{eq:empty_P} and \eqref{eq:sure_P}, we obtain
\begin{equation} \label{eq:sum1_P}
\sum_{k=1}^K \Pr(\bP\in I_k) = 1, 
\end{equation}
\begin{equation} \label{eq:LTP_P}
  \Pr(\Dpa) \;=\; \sum_{k=1}^{K} \Pr(\Dpa \mid \bP\in I_k)\, \Pr(\bP\in I_k) .
\end{equation}
Equation~\eqref{eq:LTP_P} states the law of total probability for this model.
The sum is taken over bins with positive probability.

Figure~\ref{fig:DPM} (b) illustrates the proposed DPM: the participant pool is split into intervals $I_k$ with weights $\Pr(\bP\in I_k)$, and each bin maps into the two actions—defection via $\Pr(\Dpa\mid \bP\in I_k)$ and cooperation via $\Pr(\Cpa\mid \bP\in I_k)$. Given \eqref{eq:empty_P} and \eqref{eq:sure_P}, this construction yields the decomposition in \eqref{eq:LTP_P}.

We next compare the new DPM with the conventional DPM. The conventional DPM can be viewed as a special case of the present model. Since the events $\Dop$ and $\Cop$ correspond to participants who are certain of defection and certain of cooperation, respectively, in the present parametrization $\bP$, they align with the extreme values $\bP\approx 1$ and $\bP\approx 0$. Thus the conventional DPM is recovered when all probability mass is concentrated at the extremes: 
\begin{equation}
  \Pr(\bP\in I_k)=0 \quad (k\neq 1,K).
\end{equation}

Equivalently, the conventional model rules out participants whose expectations lie strictly between the two extremes. In this sense, our model extends existing approaches by relaxing the ``\emph{certainty-only}'' restriction.

The proposed DPM is informative because it shows how the disjunction-effect pattern can be handled within classical probability. Our model does not contradict the PD findings: existing disjunction-effect experiments report overall choice rates under three information conditions (unknown / told defect / told cooperate) but do not measure $\Pr(\Dpa\mid \bP\in I_k)$ for intermediate expectation levels ($k\neq 1,K$). Accordingly, the observed data are compatible with our formulation. A number of articles have argued that classical probability is too restrictive to accommodate the experimentally observed disjunction effect (e.g., \cite{Wang_2013,Blutner_2016,Tesar_2020}). In contrast, our model shows that the disjunction effect can be accommodated within a classical framework once ambiguous expectations are allowed. Moreover, assumptions \eqref{eq:empty_P} and \eqref{eq:sure_P} are not far-fetched; they are arguably more natural than \eqref{eq:empty_DC} and \eqref{eq:sure_DC}, since people often decide under ambiguous or mixed expectations about the opponent's action.

\section{Quantum-like decision-process model and relation to the proposed classical model}
\label{section:quantum}

The disjunction effect has often been regarded as paradoxical from the standpoint of classical probability, which has motivated quantum-like models based on Hilbert-space event semantics and a quantum analogue of the law of total probability (see, e.g., \cite{Busemeyer_2006,Pothos_2009}; see also \cite{Khrennikov_2015, Pothos_2022} for reviews). In this paper we have already shown that a classical decision-process model that permits ambiguous expectations can accommodate the observed pattern (Section \ref{section:classic}). In this section we introduce the quantum formulation, interpret it as a decision-process model, and place it alongside our proposed classical model, thereby clarifying both the common goal of modeling decisions under ambiguity and the substantive differences that follow from their respective semantics and mathematics. We focus here on a widely used Hilbert-space model of disjunction-effect experiments (in particular, the PD experiment), while noting that other quantum-like models explaining disjunction-effect data have also been proposed; for an overview, see Moreira and Wichert \cite{MoreiraWichert_2016}.

\subsection{Mathematical preliminaries and quantum analogue of the law of total probability}
\label{subsec:quantum_review}

We first review the derivation of the quantum analogue of the LTP. Within the quantum framework, the state of a system is represented by a unit vector $\psi$ in a Hilbert space $\mathscr{H}$. An event $E$ is represented by a closed subspace $\mathscr{H}_E \subset \mathscr{H}$; the notion of “event” in the quantum setting differs from the classical one, and we return to this conceptual point in Section \ref{subsec:quantum_DPM}. For any closed subspace $\mathscr{H}_E$, there exists a unique orthogonal projector $\hat{\mathcal{P}}_E$ onto it. By Born's rule, the probability of $E$ in state $\psi$ is
\begin{equation} \label{eq:Born_rule}
\Pr(E) \;=\; \big\| \hat{\mathcal{P}}_{E}\,\psi \big\|^2 .
\end{equation}
The conditional probability of $E$ given $F$ is given by Lüders' rule:
\begin{equation} \label{eq:Luder_rule}
\Pr(E \mid F) \;=\; \frac{\big\| \hat{\mathcal{P}}_{E}\,\hat{\mathcal{P}}_{F}\,\psi \big\|^2}{\big\| \hat{\mathcal{P}}_{F}\,\psi \big\|^2}
\qquad \text{whenever } \Pr(F)>0 .
\end{equation}

We now represent the quantum counterparts of the opponent-choice events, $\Dop$ and $\Cop$, by closed subspaces $\mathscr{H}_{\mathrm{D}}, \mathscr{H}_{\mathrm{C}} \subset \mathscr{H}$. 
In brief, $\mathscr{H}_{\mathrm{D}}$ is the subspace representing the “opponent defects” event in the quantum model, and $\mathscr{H}_{\mathrm{C}}$ represents “opponent cooperates”; how these quantum events relate to their classical counterparts will be discussed in Section \ref{subsec:quantum_DPM}.

Assumptions analogous to \eqref{eq:empty_DC} and \eqref{eq:sure_DC} are imposed by assuming that these subspaces are orthogonal and jointly span $\mathscr{H}$:
\begin{equation}\label{eq:subspace_perp}
\mathscr H_{\mathrm{D}} \perp \mathscr H_{\mathrm{C}},
\end{equation}
\begin{equation}\label{eq:subspace_direct_sum}
\mathscr H_{\mathrm{D}} \oplus \mathscr H_{\mathrm{C}} \;=\; \mathscr H .
\end{equation}
Let $\hat{\mathcal P}_{\mathrm D}$ and $\hat{\mathcal P}_{\mathrm C}$ be the orthogonal projectors onto these subspaces. Then \eqref{eq:subspace_perp}, \eqref{eq:subspace_direct_sum} are equivalent to
\begin{equation}\label{eq:projection_perp}
\hat{\mathcal{P}}_{\mathrm{D}}\,\hat{\mathcal{P}}_{\mathrm{C}} \;=\; \hat{0},
\end{equation}
\begin{equation}\label{eq:projection_sum}
\hat{\mathcal{P}}_{\mathrm{D}} \,+\, \hat{\mathcal{P}}_{\mathrm{C}} \;=\; \hat{1}.
\end{equation}

Let $\hat{\mathcal{Q}}$ denote the projector onto the choice event $\Dpa$. From \eqref{eq:Born_rule},
\begin{equation}
\Pr(\Dpa) \;=\; \big\| \hat{\mathcal{Q}}\,\psi \big\|^2 .
\end{equation}
Using assumptions \eqref{eq:projection_perp} and \eqref{eq:projection_sum}, we obtain
\begin{equation}
\big\| \hat{\mathcal{Q}}\,\psi \big\|^2
\;=\; \big\| \hat{\mathcal{Q}}\,(\hat{\mathcal{P}}_{\mathrm{D}}+\hat{\mathcal{P}}_{\mathrm{C}})\,\psi \big\|^2
\;=\; \big\| \hat{\mathcal{Q}} \hat{\mathcal{P}}_{\mathrm{D}} \psi + \hat{\mathcal{Q}} \hat{\mathcal{P}}_{\mathrm{C}} \psi \big\|^2
\;=\; \big\| \hat{\mathcal{Q}} \hat{\mathcal{P}}_{\mathrm{D}} \psi \big\|^2 + \big\| \hat{\mathcal{Q}} \hat{\mathcal{P}}_{\mathrm{C}} \psi \big\|^2 + \Delta,
\end{equation}
where $\Delta$ is defined as
\begin{equation}
\Delta \;=\; \big\langle \psi,\; \bigl( \hat{\mathcal{P}}_{\mathrm{C}} \hat{\mathcal{Q}} \hat{\mathcal{P}}_{\mathrm{D}}
\,+\, \hat{\mathcal{P}}_{\mathrm{D}} \hat{\mathcal{Q}} \hat{\mathcal{P}}_{\mathrm{C}}\bigr) \;\psi \big\rangle .
\end{equation}
Using \eqref{eq:Born_rule} and \eqref{eq:Luder_rule}, we arrive at
\begin{equation}\label{eq:quantum_LTP}
\Pr(\Dpa) \;=\; \Pr(\Dpa \mid \Dop)\,\Pr(\Dop) \;+\; \Pr(\Dpa \mid \Cop)\,\Pr(\Cop) \;+\; \Delta .
\end{equation}
This is the quantum analogue of the LTP.
The difference from \eqref{eq:LTP_DC} is the presence of the interference term $\Delta$, which allows the defection rate in the unknown condition to fall below both informed rates.

As an important corollary of assumptions \eqref{eq:projection_perp} and \eqref{eq:projection_sum}, we obtain the following normalization identity for the opponent-choice events:
\begin{equation}\label{eq:q_norm_DC}
\Pr(\Dop)+\Pr(\Cop)
= \|\hat{\mathcal P}_{\mathrm D}\psi\|^2+\|\hat{\mathcal P}_{\mathrm C}\psi\|^2
= \big\langle \psi,(\hat{\mathcal P}_{\mathrm D}+\hat{\mathcal P}_{\mathrm C})\psi \big\rangle
= \langle \psi,\psi\rangle
= 1.
\end{equation}
This identity is highly suggestive: in the quantum model the probabilities assigned to the two opponent-choice events always sum to one for any state $\psi$. We will use \eqref{eq:q_norm_DC} in the next subsection to contrast the quantum reading of events with the proposed classical DPM.

\subsection{Interpreting the quantum-like decision-process model}
\label{subsec:quantum_DPM}

In Section \ref{subsec:quantum_review}, we reviewed the mathematical side of the quantum-like approach. We now interpret that formalism as a DPM and clarify how the quantum-like DPM relates to, and differs from, the proposed classical DPM.

We first outline how the notions of events and probabilities differ in classical and quantum frameworks. In the classical setting, events are facts about the world whose truth values are independent of any act of measurement, and probabilities quantify degrees of possibility for those facts. In the quantum setting, by contrast, events are essentially tied to measurement operations and understood as the outcomes of specified measurements (``experimental propositions'' in the sense of Birkhoff and von Neumann \cite{Birkhoff_1936}); the relevant probabilities are the chances of those outcomes when the corresponding measurement is performed. In cognitive applications, the measurement is typically instantiated by a question posed to a participant, and the event is the outcome of that question (see, e.g., \cite{PothosBusemeyer_2013}). In light of this, the proposed classical DPM models the decision process independently of whether such a question is asked, whereas the quantum-like DPM models what happens when the question is asked.

With this in view, we read the opponent-choice events $\Dop$ and $\Cop$ as the two possible answers to the question ``Do you think the opponent will defect?''—a ``defect'' answer instantiates $\Dop$ and a ``cooperate'' answer instantiates $\Cop$. We use the defection probabilities conditional on these hypothetical answers to model the corresponding rates in the two informed conditions. As shown in Section \ref{subsec:quantum_review}, equation \eqref{eq:q_norm_DC} holds for any state $\psi$. This equation describes the binary nature of the question: one of the two answers obtains when the measurement is made. Although this identity is formally the same as equation \eqref{eq:sum1_DC} in the conventional classical DPM, its interpretation differs in the quantum setting because it is grounded in measurement outcomes rather than measurement-independent facts.

Crucially, the quantum-like DPM does not assume that participants hold a fixed, settled prediction about the opponent prior to being asked. In the experiment, participants are not actually asked to state an expectation about the opponent's move. We can model their pre-decision cognitive state as a superposition $\psi=\psi_{\mathrm D}+\psi_{\mathrm C}$ relative to the ``defect/cooperate'' decomposition. This superposed state encodes an ambiguous expectation that is resolved only upon the corresponding measurement.

\subsection{Comparing the quantum-like and proposed classical models}

We are now in a position to compare the quantum-like DPM with the proposed classical DPM, highlighting both commonalities and differences. Both frameworks allow decisions to be made while the participant's expectation about the opponent's move remains ambiguous. In other words, the pre-decision state need not collapse to a determinate ``defect'' or ``cooperate'' expectation. This is the key departure from the conventional classical DPM. 

The central difference rests on how each framework defines events and assigns probabilities, as outlined in Section \ref{subsec:quantum_DPM}. In the proposed classical DPM, events are treated as measurement-independent facts, and the normalization identity \eqref{eq:sum1_P} expresses a partition of the participant pool across expectation levels. In the quantum-like DPM, events are outcomes of a specified measurement, and the normalization identity \eqref{eq:q_norm_DC} reflects the fact that a binary question yields exactly one of two outcomes when posed. 
Notably, the normalization identities make explicit—at the mathematical level—the conceptual differences in what counts as an event and how probabilities are understood.

With these interpretations in place, the models also diverge in how they formalize ambiguity. 
The proposed classical DPM captures heterogeneity between participants by means of a distribution over a continuous expectation parameter $\bP\in[0,1]$. 
An individual participant's expectation is classified as ambiguous when $\bP$ falls in an interior bin. 
By contrast, the quantum-like DPM represents ambiguity within a single cognitive state as a superposition over the two expectation subspaces (``defect'' vs.\ ``cooperate'') and includes the interference term in \eqref{eq:quantum_LTP}. 
In principle, this allows one to fit the aggregate disjunction-effect pattern without invoking heterogeneity across participants \cite{Wang_2013} (although mixed-state extensions that reintroduce such heterogeneity have also been suggested in the literature \cite{Khrennikov_2009, Blutner_2016}). 

In short, the two frameworks do not just use different mathematics; they model different notions of event and probability, and this drives their distinct treatments of ambiguity. For a compact side-by-side summary, see Table \ref{tab:model_comparison_compact}.

Some studies have compared the conventional classical DPM with the quantum-like DPM through the lens of formal event representation \cite{Widdows_2021,Widdows_2023}, and this perspective is also helpful for a comparison that includes our proposed model. Section~\ref{sec:semantics_comparison} provides a complementary, mathematically oriented comparison of the three models based on formal event semantics (set-theoretic semantics vs.\ Hilbert-space semantics).

\begin{table}[H]
\centering
\footnotesize 
\renewcommand{\arraystretch}{1.8} 
\caption{Comparison of three decision-process models (DPMs): the conventional classical DPM, the proposed classical DPM, and the quantum-like DPM.}

\newcolumntype{L}{>{\raggedright\arraybackslash}X}
\begin{tabularx}{\textwidth}{@{} >{\bfseries\raggedright}p{0.18\textwidth} L L L @{}}
\toprule
& \textbf{Conventional classical DPM} & \textbf{Proposed classical DPM} & \textbf{Quantum-like DPM} \\
\midrule 

Assumption
& $\Dop \land \Cop=\varnothing$ \par $\Dop \lor \Cop=\top$ \par \eqref{eq:empty_DC}, \eqref{eq:sure_DC}
& $[\bP\in I_k]\land[\bP\in I_\ell]=\varnothing$ \par $(\text{for} \;k\neq \ell)$ \par $\displaystyle\bLor_{k=1}^{K}[\bP\in I_k]=\top$ \par \eqref{eq:empty_P}, \eqref{eq:sure_P}
& $\mathscr H_{\mathrm D}\!\perp\!\mathscr H_{\mathrm C}$ \par $\mathscr H_{\mathrm D}\!\oplus\!\mathscr H_{\mathrm C}=\mathscr H$ \par \eqref{eq:subspace_perp}, \eqref{eq:subspace_direct_sum} \\
\midrule 

Treatment of participants' expectations
& Certainty-only \par (defect or cooperate)
& Allow mixed expectations via $\bP\!\in\![0,1]$
& Allow mixed expectations via superposition \\
\midrule 

Law of total probability
& \footnotesize
$\Pr(\Dpa)=
\begin{aligned}[t]
  &\Pr(\Dpa\mid\Dop)\Pr(\Dop)\\
  &+\Pr(\Dpa\mid\Cop)\Pr(\Cop)
\end{aligned}$ \par \eqref{eq:LTP_DC}
& \footnotesize
$\Pr(\Dpa)=
\begin{aligned}[t]
  &\textstyle\sum_{k=1}^{K} \Pr(\Dpa\mid\bP\!\in\! I_k)\\
  &\cdot \Pr(\bP\!\in\! I_k)
\end{aligned}$ \par \eqref{eq:LTP_P}
& \footnotesize
$\Pr(\Dpa)=
\begin{aligned}[t]
  &\Pr(\Dpa\mid\Dop)\Pr(\Dop)\\
  &+\Pr(\Dpa\mid\Cop)\Pr(\Cop)\\
  &+\Delta
\end{aligned}$ \par \eqref{eq:quantum_LTP}\\

Can the unknown-condition rate fall below both informed rates?
& No (conflict)
& Yes (interior terms $\Pr(\Dpa\mid\bP\!\in\! I_k)$ not fixed by data)
& Yes (if $\Delta$ is sufficiently negative) \\
\midrule 

Normalization
& $\Pr(\Dop)+\Pr(\Cop)=1$ \par \eqref{eq:sum1_DC}
& $\displaystyle \sum_{k=1}^{K}\Pr(\bP\in I_k)=1$ \par \eqref{eq:sum1_P}
& $\Pr(\Dop)+\Pr(\Cop)=1 $ \par \eqref{eq:q_norm_DC} \\
\bottomrule
\end{tabularx}

\label{tab:model_comparison_compact}
\end{table}

\newpage
\section{Mathematical Results}\label{mathresults}

Here we prove two mathematical results to support the above arguments. 
The first result is the following Theorem~\ref{thm:disjunction_realization}, which makes precise, in a fully Kolmogorovian setting, the above claim that the proposed classical decision-process model can reproduce the empirical disjunction-effect pattern in the Prisoner's Dilemma while strictly obeying the classical law of total probability. 
We refer to each interval $I_k$ in the partition of $[0,1]$ defined in Section~\ref{section:classic} as a \emph{bin}.

\begin{thm}[Classical realization of a strong disjunction triple]
\label{thm:disjunction_realization}
Let 
\[
d_U,d_D,d_C \in (0,1),
\] 
interpreted as the observed defection
rates in the Prisoner's Dilemma experiment under the ``unknown'', ``told
defect'', and ``told cooperate'' conditions, respectively. Then there exists
a three-bin instance of the classical decision-process model of
Section~\ref{section:classic}, with expectation parameter $\bP\in[0,1]$,
partition $\{I_1,I_2,I_3\}$ of $[0,1]$, bin weights
\[
\alpha_k := \Pr(\bP\in I_k) \quad (k=1,2,3),
\]
and bin-wise defection probabilities
\[
q_k := \Pr(\Dpa \mid \bP\in I_k) \quad (k=1,2,3),
\]
such that:
\begin{enumerate}
  \item $\alpha_k>0$ for all $k$ and $\sum_{k=1}^3 \alpha_k = 1$;
  \item the extreme bins correspond to the ``told cooperate'' and ``told
        defect'' conditions,
        \[
          q_1 = d_C = \Pr(\Dpa\mid\Cop),\qquad
          q_3 = d_D = \Pr(\Dpa\mid\Dop),
        \]
        where we identify $\Cop := [\bP\in I_1]$ and $\Dop := [\bP\in I_3]$;
  \item the ``unknown'' condition satisfies the classical law of total
        probability in the three-bin form
        \[
          \Pr(\Dpa) \;=\; \sum_{k=1}^3 \Pr(\Dpa\mid \bP\in I_k)\,
                                   \Pr(\bP\in I_k)
                       \;=\; \sum_{k=1}^3 q_k \alpha_k
                       \;=\; d_U.
        \]
\end{enumerate}
In particular, such a model exists even when $d_U < \min\{d_D,d_C\}$, i.e.,
in the strong disjunction-effect case.
\end{thm}

\begin{proof}
Fix $d_U,d_D,d_C\in(0,1)$. Take $K=3$ and choose any partition
\[
0 = a_0 < a_1 < a_2 < a_3 = 1,\qquad
I_1=[0,a_1),\quad I_2=[a_1,a_2),\quad I_3=[a_2,1].
\]
We interpret $I_1$ and $I_3$ as high-cooperation and high-defection
expectation regions, respectively, and $I_2$ as an intermediate region, as in
Figure~\ref{fig:DPM}(b).

\medskip
\noindent\textit{Step 1: bin probabilities.}
Pick $t>0$ such that
\[
  0 < t <
  \min\Bigl\{\frac{1}{2},\frac{d_U}{2},\frac{1-d_U}{2}\Bigr\}.
\]
Define
\[
  \alpha_1 := t,\qquad
  \alpha_3 := t,\qquad
  \alpha_2 := 1-2t.
\]
Then $\alpha_k>0$ for all $k$ and $\sum_{k=1}^3 \alpha_k=1$.

\medskip
\noindent\textit{Step 2: extreme bins match the known conditions.}
Set
\[
  q_1 := d_C,\qquad q_3 := d_D.
\]
If we identify $\Cop := [\bP\in I_1]$ and $\Dop := [\bP\in I_3]$, then
\[
  \Pr(\Dpa\mid\Cop) = q_1 = d_C,\qquad
  \Pr(\Dpa\mid\Dop) = q_3 = d_D,
\]
so the ``told cooperate'' and ``told defect'' rates are matched exactly in the
three-bin model.

\medskip
\noindent\textit{Step 3: solving for the interior bin.}
We now choose $q_2$ so that the unconditional defection rate in the
``unknown'' condition equals $d_U$. Under \eqref{eq:LTP_P} for $K=3$ we
require
\[
  d_U
  = \Pr(\Dpa)
  = \sum_{k=1}^3 q_k \alpha_k
  = t\,d_C + (1-2t)\,q_2 + t\,d_D.
\]
Solving for $q_2$ gives
\begin{equation}\label{eq:q2_formula}
  q_2
  = \frac{d_U - t(d_C + d_D)}{1-2t}.
\end{equation}
The denominator is strictly positive because $t<1/2$. We claim that
$0<q_2<1$.

For the lower bound, use $d_C,d_D\le 1$ to obtain
\[
  t(d_C+d_D) \le 2t
  \quad\Rightarrow\quad
  d_U - t(d_C+d_D) \ge d_U - 2t > 0,
\]
where the last inequality is $t<d_U/2$. Dividing by $1-2t>0$ yields
$q_2>0$.

For the upper bound, $d_C,d_D\ge 0$ implies
\[
  d_U - t(d_C+d_D) \le d_U.
\]
Since $t<(1-d_U)/2$, we have $d_U < 1-2t$, so
\[
  d_U - t(d_C+d_D) < 1-2t,
\]
and hence $q_2<1$ after division by $1-2t$. Thus $q_2\in(0,1)$.

\medskip
\noindent\textit{Step 4: explicit probability space.}
Define\footnote{Note that the first coordinate is the expectation bin, the second is the player's action.}
\[
  \Omega := \{1,2,3\}\times\{\mathrm{D},\mathrm{C}\},
\]
and a probability measure $\Pr$ on singletons by
\[
  \Pr\bigl((k,\mathrm{D})\bigr) := \alpha_k q_k,\qquad
  \Pr\bigl((k,\mathrm{C})\bigr) := \alpha_k (1-q_k),
  \quad k=1,2,3,
\]
extending to all subsets of $\Omega$ by finite additivity. This is a
probability measure, since
\[
  \sum_{k=1}^3 \Bigl[\Pr\bigl((k,\mathrm{D})\bigr)
                     + \Pr\bigl((k,\mathrm{C})\bigr)\Bigr]
  = \sum_{k=1}^3 \bigl[\alpha_k q_k + \alpha_k(1-q_k)\bigr]
  = \sum_{k=1}^3 \alpha_k
  = 1.
\]

Define the expectation parameter $\bP\colon\Omega\to[0,1]$ by setting
$\bP(\omega)\in I_k$ whenever the first coordinate of $\omega$ equals $k$.
Then $[\bP\in I_k]=\{k\}\times\{\mathrm{D},\mathrm{C}\}$ and
\[
  \Pr(\bP\in I_k) = \alpha_k,\qquad
  \Pr(\Dpa\mid\bP\in I_k) = q_k,\quad k=1,2,3.
\]
Let
\[
  \Dpa := \{(k,\mathrm{D}):k=1,2,3\},\qquad
  \Cop := [\bP\in I_1],\qquad
  \Dop := [\bP\in I_3].
\]
Then, by construction,
\[
  \Pr(\Dpa\mid\Cop) = q_1 = d_C,\qquad
  \Pr(\Dpa\mid\Dop) = q_3 = d_D,
\]
and, using \eqref{eq:LTP_P} in the three-bin form,
\[
  \Pr(\Dpa)
  = \sum_{k=1}^3 \Pr(\Dpa\mid\bP\in I_k)\,\Pr(\bP\in I_k)
  = \sum_{k=1}^3 q_k \alpha_k
  = d_U.
\]
Thus the triple $(d_U,d_D,d_C)$ is realized by a three-bin instance of the
proposed classical DPM that satisfies the classical law of total probability
\eqref{eq:LTP_P}, including the strong disjunction-effect case
$d_U<\min\{d_D,d_C\}$. 
\end{proof}

Thus the theorem shows that the strong inequality $d_U < \min\{d_D,d_C\}$ is \emph{not} in conflict with classical Kolmogorov probability or with the law of total probability. What fails is only the overly restrictive ``certainty-only'' classical DPM, which collapses all participants into two
extreme mental states and thereby forbids any intermediate expectation levels. It demonstrates that once ambiguous expectations are explicitly modeled (via a latent expectation parameter $\bP$ and at least one interior bin), a classical model that respects the LTP can always be found that matches the data exactly.

Conceptually, this theorem has two important consequences for the wider discussion:
\begin{enumerate}
\item It undercuts the common inference from the PD data to a violation of the classical LTP: the conflict is traced to a modeling choice (the certainty-only premise), not to classical probability itself.
\item It shows that a very low-complexity classical model (a single latent continuous expectation parameter with only three bins) already has enough expressive power to capture the observed pattern, without resorting to elaborate latent-variable networks or Markov dynamics.
\end{enumerate}

The following Theorem~\ref{cor:quantum_disjunction_realization} lifts
Theorem~\ref{thm:disjunction_realization} to a direct comparison between the classical and quantum-like decision-process models.

\begin{thm}[Classical realization of quantum disjunction probabilities]
\label{cor:quantum_disjunction_realization}
Consider the quantum-like decision-process model for the Prisoner's Dilemma
disjunction-effect experiment described in Section~\ref{section:quantum}, with Hilbert space
$\mathscr H$, unit state vector $\psi\in\mathscr H$, opponent-choice subspaces
$\mathscr{H}_\mathrm{D},\mathscr{H}_\mathrm{C} \subseteq\mathscr H$ with corresponding orthogonal projectors
$\hat{\mathcal{P}}_{\mathrm{D}},\hat {\mathcal{P}}_{\mathrm{C}}$, and a projector $\hat {\mathcal{Q}}$ representing the choice event
$\Dpa$.

Let $\Pr_{\mathrm{qm}}$ denote the probability assignment defined from
$\psi$ by the Born rule and Lüders conditionalization in \eqref{eq:Born_rule} and \eqref{eq:Luder_rule}, and
define the three defection rates
\[
  d_U := \Pr_{\mathrm{qm}}(\Dpa),\qquad
  d_D := \Pr_{\mathrm{qm}}(\Dpa \mid \Dop),\qquad
  d_C := \Pr_{\mathrm{qm}}(\Dpa \mid \Cop),
\]
where $\Dop$ and $\Cop$ are the opponent-defect and opponent-cooperate
events represented by $\mathscr{H}_\mathrm{D}$ and $\mathscr{H}_\mathrm{C}$, respectively.
Assume  $\Pr_{\mathrm{qm}}(\Dop)>0$, $\Pr_{\mathrm{qm}}(\Cop)>0$ so that the
conditional probabilities are well-defined, and that
\[
  d_U,d_D,d_C \in (0,1).
\]

Then there exists a three-bin instance of the classical decision-process model
of Section~\ref{section:classic}, with expectation parameter $\bP\in[0,1]$, a
partition $\{I_1,I_2,I_3\}$ of $[0,1]$, bin weights
\[
  \alpha_k := \Pr_{\mathrm{cl}}(\bP\in I_k) \quad (k=1,2,3),
\]
and bin-wise defection probabilities
\[
  q_k := \Pr_{\mathrm{cl}}(\Dpa \mid \bP\in I_k) \quad (k=1,2,3),
\]
for some classical probability measure $\Pr_{\mathrm{cl}}$, such that, with the
identifications
\[
  \Cop := [\bP\in I_1],\qquad
  \Dop := [\bP\in I_3],
\]
the three classical defection rates coincide with the quantum ones:
\[
  \Pr_{\mathrm{cl}}(\Dpa) = d_U,\qquad
  \Pr_{\mathrm{cl}}(\Dpa \mid \Dop) = d_D,\qquad
  \Pr_{\mathrm{cl}}(\Dpa \mid \Cop) = d_C.
\]
In particular, any strong disjunction-effect pattern
$d_U < \min\{d_D,d_C\}$ produced by the quantum-like model is also realized
by a three-bin instance of the classical decision-process model.
\end{thm}

\begin{proof}
By construction, the quantum-like model assigns well-defined numbers
$d_U,d_D,d_C\in(0,1)$ via
\[
  d_U = \Pr_{\mathrm{qm}}(\Dpa),\qquad
  d_D = \Pr_{\mathrm{qm}}(\Dpa \mid \Dop),\qquad
  d_C = \Pr_{\mathrm{qm}}(\Dpa \mid \Cop),
\]
using the Born rule and Lüders conditionalization in \eqref{eq:Born_rule} and \eqref{eq:Luder_rule}.

Apply Theorem~\ref{thm:disjunction_realization} to the triple
$(d_U,d_D,d_C)$. The theorem yields a three-bin instance of the classical
decision-process model of Section~\ref{section:classic}, with partition
$\{I_1,I_2,I_3\}$ of $[0,1]$, bin weights
$\alpha_k := \Pr_{\mathrm{cl}}(\bP\in I_k)>0$ satisfying
$\sum_{k=1}^3 \alpha_k = 1$, and bin-wise defection probabilities
$q_k := \Pr_{\mathrm{cl}}(\Dpa \mid \bP\in I_k)\in(0,1)$ such that
\[
  q_1 = d_C,\qquad
  q_3 = d_D,\qquad
  \sum_{k=1}^3 q_k \alpha_k = d_U,
\]
and where we identify
\[
  \Cop := [\bP\in I_1],\qquad
  \Dop := [\bP\in I_3].
\]

By the definition of $q_1$ and $q_3$ we have
\[
  \Pr_{\mathrm{cl}}(\Dpa \mid \Cop)
  = \Pr_{\mathrm{cl}}(\Dpa \mid \bP\in I_1)
  = q_1
  = d_C,
\]
and similarly
\[
  \Pr_{\mathrm{cl}}(\Dpa \mid \Dop)
  = \Pr_{\mathrm{cl}}(\Dpa \mid \bP\in I_3)
  = q_3
  = d_D.
\]
Moreover, the three-bin law of total probability~\eqref{eq:LTP_P} in the
classical model gives
\[
  \Pr_{\mathrm{cl}}(\Dpa)
  = \sum_{k=1}^3
      \Pr_{\mathrm{cl}}(\Dpa \mid \bP\in I_k)\,
      \Pr_{\mathrm{cl}}(\bP\in I_k)
  = \sum_{k=1}^3 q_k \alpha_k
  = d_U.
\]
Hence the classical model guaranteed by Theorem~\ref{thm:disjunction_realization}
reproduces exactly the three defection rates of the quantum-like model.
If, in addition, the quantum triple satisfies the strong disjunction-effect
inequality $d_U < \min\{d_D,d_C\}$, then the same inequality holds for the
corresponding probabilities in the classical model, since the two triples
coincide. This establishes the claim.
\end{proof}

Put another way, at the level of the three observable defection rates, the
quantum-like model has no greater empirical representational power than
the three-bin classical model: whatever triple the quantum model can produce
under its assumptions, there is a classical model (obeying the ordinary LTP)
that produces exactly the same triple. In particular, any strong
disjunction-effect pattern generated by the quantum-like model is also
generated by an appropriate classical model.

This has two main implications:
\begin{enumerate}
  \item The interference term $\Delta$ in the quantum LTP \eqref{eq:quantum_LTP} is not forced
        by the PD data themselves; the same aggregate pattern can be
        accounted for by a classical mixture over expectation levels. Quantum
        structure is therefore not required merely to fit the three
        observed rates.
  \item The real distinction between the proposed classical DPM and the
        quantum-like DPM lies in their event semantics and in where
        ambiguity is represented: heterogeneity across participants via the
        distribution of $\bP$ in the classical model versus superposition
        within a single cognitive state in the quantum model. The theorem shows that,
        when we set these semantic differences aside and consider only the three
        observable rates, the proposed classical model can exactly match any
        triple $(d_U,d_D,d_C)\in(0,1)^3$ generated by the quantum-like model.
\end{enumerate}

Taken together, Theorem~\ref{thm:disjunction_realization} and
Theorem~\ref{cor:quantum_disjunction_realization} therefore sharpen 
the present paper's central point: the disjunction-effect pattern in the 
Prisoner's Dilemma does not by itself refute classical probability, and 
quantum-like formalisms owe their success at least as much to how they 
represent ambiguous expectations as to any inherently non-classical 
probability structure.

\section{Event semantics and how ambiguity enters the three models}
\label{sec:semantics_comparison}

This section complements the argument in Section \ref{section:quantum} by comparing the event semantics used in the conventional classical DPM, the proposed classical DPM, and the quantum-like DPM. Our aim here is to clarify why the latter two models allow participants to hold ambiguous pre-decision expectations about the opponent's choice, whereas the conventional model precludes such ambiguity. For ease of comparison, we use the standard interpretation brackets $\llbracket\cdot\rrbracket$ and write $\llbracket E\rrbracket$ for the semantic denotation of an event $E$. In the classical readings, $\llbracket E\rrbracket$ is a subset of a sample space $\Omega$. In the quantum-like reading, $\llbracket E\rrbracket$ is a closed subspace of a Hilbert space $\mathscr H$. 
Conjunction ($\land$) is interpreted as intersection in both settings. Disjunction ($\lor$) is interpreted as set-theoretic union in the classical setting, but as the closed linear span in the quantum setting (the smallest closed subspace containing both event subspaces). Negation ($\lnot$) is interpreted as set complement classically, $\llbracket \lnot E \rrbracket =\Omega\setminus \llbracket E \rrbracket$, and as orthogonal complement in the quantum setting, $\llbracket \lnot E \rrbracket={\llbracket E \rrbracket}^{\perp}$ (in $\mathscr H$).

This Hilbert-space semantics goes back to the standard formulation of quantum logic by Birkhoff and von Neumann \cite{Birkhoff_1936}, has been introduced into cognitive modeling as a non-classical event semantics (see, e.g., \cite{Wang_2013, Khrennikov_2015}), and has recently been illustrated geometrically and applied to decision-making phenomena such as the disjunction effect by Widdows et al.\ \cite{Widdows_2021,Widdows_2023}.

\paragraph{Conventional DPM (classical, Boolean semantics)}
Opponent-choice events are interpreted as subsets of the sample space $\Omega$ (with $\llbracket \top\rrbracket=\Omega$):
$\llbracket \Dop \rrbracket=\Omega_{\mathrm D},\;
\llbracket \Cop \rrbracket=\Omega_{\mathrm C}$, where $\Omega_{\mathrm D}, \Omega_{\mathrm C}\subset \Omega$.
Assumptions \eqref{eq:empty_DC} and \eqref{eq:sure_DC} are interpreted as
\begin{equation}
\llbracket \Dop\rrbracket \cap \llbracket \Cop\rrbracket = \varnothing,
\qquad
\llbracket \Dop\rrbracket \cup \llbracket \Cop\rrbracket = \llbracket \top\rrbracket .
\end{equation}
Thus every element $\omega\in\Omega$ belongs to exactly one of the two events.
This leaves no room, at the level of event semantics, for participants who proceed with an ambiguous or mixed expectation about the opponent's action (cf.\ Section \ref{subsec:decision-process_model_review}).

\paragraph{Proposed DPM (classical, Boolean semantics with a refined partition)}
We retain the classical set-theoretic event semantics, but refine the description of pre-decision expectations by introducing a continuous expectation parameter $\bP\in[0,1]$. Let $\{I_k\}_{k=1}^K$ be a partition of $[0,1]$; this induces a partition of the sample space via the events $\llbracket \bP\in I_k\rrbracket$, namely
\begin{equation}
\llbracket \bP\in I_k \rrbracket \cap \llbracket \bP\in I_\ell \rrbracket \;=\; \varnothing \quad (k\neq \ell),
\qquad
\bigcup_{k=1}^K \llbracket \bP\in I_k \rrbracket \;=\; \llbracket \top\rrbracket,
\end{equation}
which are the semantic counterparts of assumptions \eqref{eq:empty_P} and \eqref{eq:sure_P}.
The two settled expectation events are represented by the extremes,
\begin{equation}
\llbracket \Dop\rrbracket := \llbracket \bP\in I_K \rrbracket, 
\qquad
\llbracket \Cop\rrbracket := \llbracket \bP\in I_1 \rrbracket .
\end{equation}
Because interior bins are admitted, we have
\begin{equation}
\llbracket \Dop\rrbracket \cup \llbracket \Cop\rrbracket
\;=\; \llbracket \bP\in I_1 \rrbracket \cup \llbracket \bP\in I_K \rrbracket
\;\subsetneq\;
\bigcup_{k=1}^K \llbracket \bP\in I_k \rrbracket
\;=\; \llbracket \top\rrbracket .
\end{equation}
Thus, ambiguity about the opponent's move corresponds to those sample points that lie outside both certainty sets,
\begin{equation}
\llbracket \top\rrbracket \setminus \bigl(\llbracket \Dop\rrbracket \cup \llbracket \Cop\rrbracket\bigr),
\end{equation}
namely participants whose expectation parameter falls in interior levels rather than at the two extremes.

\paragraph{Quantum-like DPM (Hilbert-space semantics)}
In the classical Boolean logic, the two assumptions used in the conventional DPM, \eqref{eq:empty_DC} and \eqref{eq:sure_DC}, are equivalent to the single syntactic identity
\begin{equation}
\Dop \equiv \lnot \Cop .
\end{equation}
The quantum-like model adopts exactly the same syntax. What changes is the semantics. In Hilbert-space semantics, events are represented by closed subspaces of $\mathscr H$: $\llbracket \Dop \rrbracket = \mathscr H_{\mathrm D} \subset \mathscr{H}, \; \llbracket \Cop \rrbracket = \mathscr H_{\mathrm C} \subset \mathscr{H}$, with $\llbracket \top \rrbracket = \mathscr{H}$.
Under this reading, the identity above is interpreted as
\begin{equation}
\llbracket \Dop \rrbracket \;=\; \llbracket \Cop \rrbracket^{\perp},
\end{equation}
which is equivalent to the pair of subspace conditions
\begin{equation}
\llbracket \Dop \rrbracket \perp \llbracket \Cop \rrbracket,
\qquad
\llbracket \Dop \rrbracket \oplus \llbracket \Cop \rrbracket \;=\; \llbracket \top \rrbracket .
\end{equation}
These are precisely \eqref{eq:subspace_perp} and \eqref{eq:subspace_direct_sum}; the second condition is the Hilbert-space semantic counterpart of the syntactic disjunction $\Dop \lor \Cop \equiv \top$ and thus provides the quantum analogue of exhaustivity.
Nevertheless, the set-theoretic union of the two subspaces is, in general,
\begin{equation}
\llbracket \Dop \rrbracket \cup \llbracket \Cop \rrbracket
= \mathscr H_{\mathrm D} \cup \mathscr H_{\mathrm C}
\subsetneq \mathscr H = \llbracket \top \rrbracket .
\end{equation}
Vectors of the form $\psi = \psi_{\mathrm D} + \psi_{\mathrm C}$ with $\psi_{\mathrm D} \in \mathscr H_{\mathrm D}$ and $\psi_{\mathrm C} \in \mathscr H_{\mathrm C}$ lie in the closed span $\overline{\mathscr H_{\mathrm D} + \mathscr H_{\mathrm C}}$ but in neither $\mathscr H_{\mathrm D}$ nor $\mathscr H_{\mathrm C}$ whenever both components are nonzero.

This illustrates the characteristic ``in-between'' states of quantum logic discussed in geometric terms by Widdows et al. \cite{Widdows_2021}, which they later use to analyze the quantum disjunction underlying the disjunction effect (Widdows et al. \cite{Widdows_2023}).

In this way, ambiguity is represented within a single state: $\psi \notin \llbracket \Dop \rrbracket \cup \llbracket \Cop \rrbracket$ even though $\psi \in \llbracket \Dop \lor \Cop \rrbracket = \llbracket \top \rrbracket$.

\paragraph{Common departure and different reasons}
Both the proposed classical DPM and the quantum-like DPM depart from the conventional cover condition
\begin{equation}
\llbracket \Dop\rrbracket \cup \llbracket \Cop\rrbracket \;=\; \llbracket \top\rrbracket
\end{equation}
in order to accommodate ambiguity:
\begin{align*}
\text{Proposed classical DPM:}& \ \ \llbracket \Dop\rrbracket \cup \llbracket \Cop\rrbracket \subsetneq \llbracket \top\rrbracket 
\ \ \text{via interior expectation bins},
\\
\text{Quantum-like DPM:}& \ \ \llbracket \Dop\rrbracket \cup \llbracket \Cop\rrbracket \subsetneq \llbracket \top\rrbracket 
\ \ \text{via superposition states}.
\end{align*}
They share the same displayed inequality but for different semantic reasons: the proposed model keeps classical set semantics and enriches the event vocabulary (adding interior expectation events), whereas the quantum-like model changes the semantic notion of disjunction from union to closed span, thereby allowing states that belong to neither sharp alternative. Figure \ref{fig:venn_diagram} illustrates these contrasts and makes explicit that the failure of the cover condition arises from distinct mechanisms in the two models.

\begin{figure}[H]
    \centering
    \includegraphics[width=1.0\linewidth]{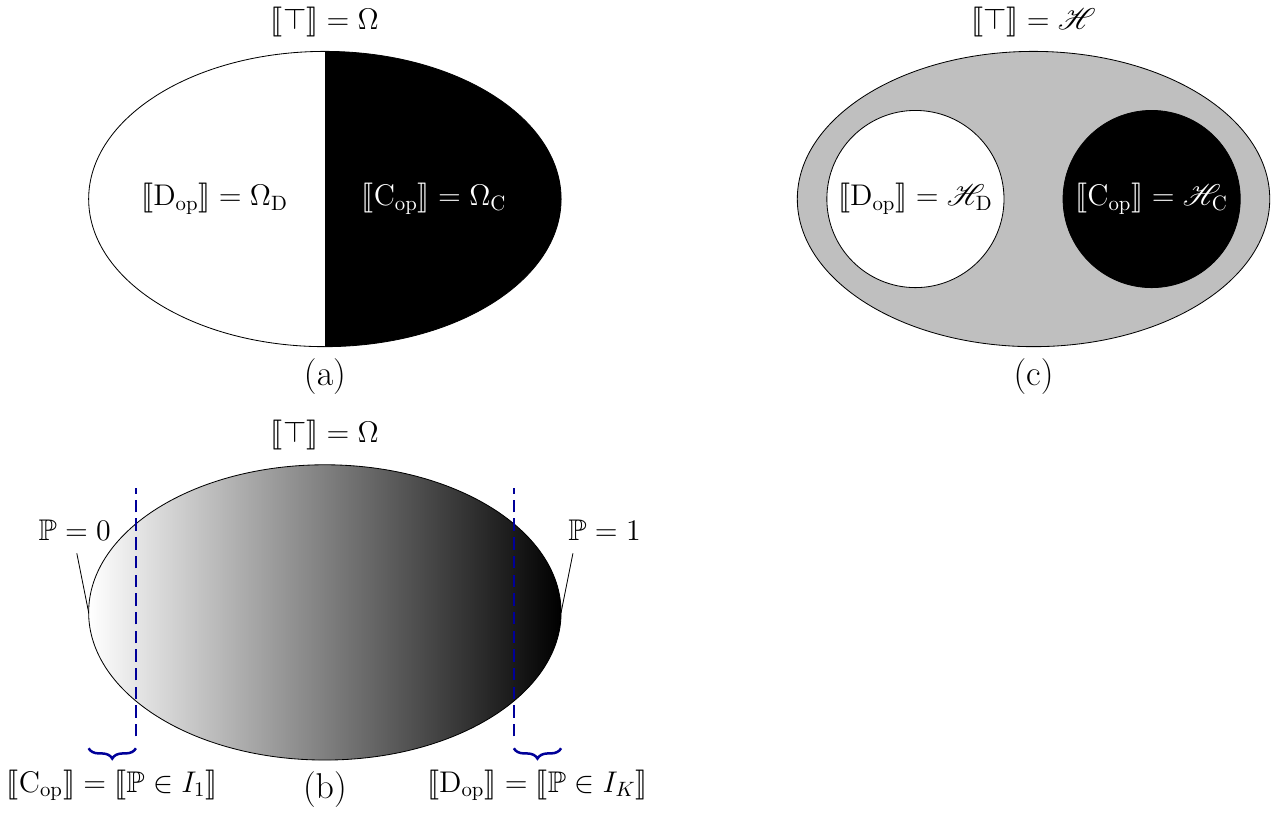}
    \caption{Event semantics across the three decision-process models. (a) Conventional classical (Boolean) reading: the certainty events exhaust the space, $\llbracket \Dop\rrbracket\cup\llbracket \Cop\rrbracket=\llbracket\top\rrbracket=\Omega$, so ambiguity has no semantic place. (b) Proposed classical reading: we keep set semantics but introduce an expectation parameter $\bP\in[0,1]$; the certainty events are the extremes, $\llbracket \Cop\rrbracket=\llbracket \bP\in I_1\rrbracket$ and $\llbracket \Dop\rrbracket=\llbracket \bP\in I_K\rrbracket$, while interior levels account for $\llbracket\top\rrbracket\setminus(\llbracket \Dop\rrbracket\cup\llbracket \Cop\rrbracket)$. (c) Quantum-like reading: events are closed subspaces $\mathscr H_{\mathrm D},\mathscr H_{\mathrm C}\subset\mathscr H$ with $\mathscr H_{\mathrm D}\oplus\mathscr H_{\mathrm C}=\llbracket\top\rrbracket=\mathscr H$; superposition states lie in $\mathscr H$ but in neither subspace, so $\llbracket \Dop\rrbracket\cup\llbracket \Cop\rrbracket\subsetneq\llbracket\top\rrbracket$. } 
    \label{fig:venn_diagram}
\end{figure}

We finally prove the following theorem to support the above arguments.

\begin{thm}[Ambiguity sets in the three decision-process models]\label{thm:ambiguity-sets}
Given opponent-choice events $\Dop$ and $\Cop$ and the sure event $\top$, define
the associated \emph{ambiguity set} by
\[
A \;:=\; \llbracket \top\rrbracket \setminus
\bigl(\llbracket \Dop\rrbracket \cup \llbracket \Cop\rrbracket\bigr).
\]

\begin{enumerate}[(i)]
\item \textbf{Conventional DPM.}
In the conventional model, assume that the opponent-choice events satisfy
\begin{equation}\label{eq:conv-partition}
\llbracket \Dop\rrbracket \cap \llbracket \Cop\rrbracket = \varnothing,
\qquad
\llbracket \Dop\rrbracket \cup \llbracket \Cop\rrbracket
= \llbracket \top\rrbracket.
\end{equation}
Then $A = \varnothing$.

\item \textbf{Proposed DPM with expectation parameter.}
In the proposed model, let $\bP\in[0,1]$ be the expectation parameter and
$\{I_k\}_{k=1}^K$ a partition of $[0,1]$. The induced events
$\llbracket \mathbb{P}\in I_k\rrbracket$ satisfy
\begin{equation}\label{eq:bins-partition}
\llbracket \mathbb{P}\in I_k\rrbracket \cap \llbracket \mathbb{P}\in I_\ell\rrbracket = \varnothing
\quad (k\neq \ell),
\qquad
\bigcup_{k=1}^K \llbracket \mathbb{P}\in I_k\rrbracket = \llbracket \top\rrbracket.
\end{equation}
The two settled expectation events are represented by the extremes
\begin{equation}\label{eq:extremes}
\llbracket \Dop\rrbracket := \llbracket \mathbb{P}\in I_K\rrbracket,
\qquad
\llbracket \Cop\rrbracket := \llbracket \mathbb{P}\in I_1\rrbracket.
\end{equation}
Then
\begin{equation}\label{eq:A-classical-refined}
A
= \bigcup_{k=2}^{K-1} \llbracket \mathbb{P}\in I_k\rrbracket.
\end{equation}
In particular, $A=\varnothing$ whenever $K\le 2$, and for $K\ge 3$ one has
\[
A\neq \varnothing
\quad\Longleftrightarrow\quad
\exists\,k\in\{2,\dots,K-1\}\ \text{such that }\ \llbracket \bP \in I_k\rrbracket\neq \varnothing.
\]
Moreover, for any probability measure $\Pr$ on $\Omega$,
\[
\Pr(A) > 0
\quad\Longleftrightarrow\quad
\exists\,k\in\{2,\dots,K-1\}\ \text{such that }\
\Pr\bigl(\bP \in I_k\bigr) > 0.
\]

\item \textbf{Quantum-like DPM.}
In the quantum-like model, let $\mathscr{H}$ be a finite-dimensional complex Hilbert space with
$\dim \mathscr{H} = n\geq 2$, and let
\[
\llbracket \Dop\rrbracket =: \mathscr{H}_{\mathrm{D}} \subset \mathscr{H},
\qquad
\llbracket \Cop\rrbracket =: \mathscr{H}_{\mathrm{C}} \subset \mathscr{H}
\]
be nonzero closed subspaces with
\begin{equation}\label{eq:orth-direct-sum}
\mathscr{H}_{\mathrm{D}} \perp \mathscr{H}_{\mathrm{C}},
\qquad
\mathscr{H}_{\mathrm{D}} \oplus \mathscr{H}_{\mathrm{C}} = \mathscr{H}.
\end{equation}
Let
\[
S(\mathscr{H}) \;:=\; \{\psi\in \mathscr{H} : \|\psi\| = 1\}
\]
be the unit sphere, and let $\mu$ denote the (unique) unitarily invariant probability
measure on $S(\mathscr{H})$. Define the \emph{state-level ambiguity set}
\[
A_S
\;:=\;
S(\mathscr{H})\setminus \bigl( (\mathscr{H}_{\mathrm{D}}\cap S(\mathscr{H})) \cup (\mathscr{H}_{\mathrm{C}}\cap S(\mathscr{H})) \bigr).
\]
Then:
\begin{enumerate}[(a)]
\item every $\psi\in S(\mathscr{H})$ admits a unique decomposition
$\psi = \psi_D + \psi_C$ with $\psi_D\in \mathscr{H}_{\mathrm{D}}$ and $\psi_C\in \mathscr{H}_{\mathrm{C}}$, and
\[
\psi\in A_S
\quad\Longleftrightarrow\quad
\psi_D\neq 0\ \text{and}\ \psi_C\neq 0;
\]
\item $A_S$ is dense in $S(\mathscr{H})$;
\item $\mu(A_S)=1$, i.e.\ $\mu\bigl((\mathscr{H}_{\mathrm{D}}\cap S(\mathscr{H}))\cup(\mathscr{H}_{\mathrm{C}}\cap S(\mathscr{H}))\bigr)=0$.
\end{enumerate}
\end{enumerate}
\end{thm}

\begin{proof}
By definition,
\[
A \;=\; \llbracket \top\rrbracket \setminus
\bigl(\llbracket \Dop\rrbracket \cup \llbracket \Cop\rrbracket\bigr).
\]

\medskip
\noindent\emph{(i) Conventional classical DPM.}
Under~\eqref{eq:conv-partition} we have
$\llbracket \Dop\rrbracket \cup \llbracket \Cop\rrbracket
= \llbracket \top\rrbracket$, so
\[
A
= \llbracket \top\rrbracket
  \setminus \bigl(\llbracket \Dop\rrbracket \cup \llbracket \Cop\rrbracket\bigr)
= \llbracket \top\rrbracket \setminus \llbracket \top\rrbracket
= \varnothing.
\]

\medskip
\noindent\emph{(ii) Proposed classical DPM with expectation parameter.}
From~\eqref{eq:bins-partition} and~\eqref{eq:extremes} we have
\[
\llbracket \top\rrbracket
= \bigcup_{k=1}^K \llbracket \bP \in I_k\rrbracket,
\qquad
\llbracket \Dop\rrbracket \cup \llbracket \Cop\rrbracket
= \llbracket \bP \in I_K\rrbracket \cup \llbracket \bP \in I_1\rrbracket.
\]
Hence
\begin{align*}
A
&= \left(\bigcup_{k=1}^K \llbracket \bP \in I_k\rrbracket\right)
   \setminus
   \bigl(\llbracket \bP \in I_1\rrbracket \cup \llbracket \bP \in I_K\rrbracket\bigr) \\
&= \bigcup_{k=2}^{K-1} \llbracket \bP \in I_k\rrbracket,
\end{align*}
where the last step uses pairwise disjointness in~\eqref{eq:bins-partition}: subtracting
the union of two disjoint members of a disjoint union removes exactly those members.

If $K\le 2$, the index set $\{2,\dots,K-1\}$ is empty and thus the union is empty,
so $A=\varnothing$. If $K\ge 3$, then~\eqref{eq:A-classical-refined} implies
\[
A\neq \varnothing
\quad\Longleftrightarrow\quad
\exists\,k\in\{2,\dots,K-1\}\ \text{with}\ \llbracket \bP \in I_k\rrbracket\neq \varnothing,
\]
since a finite union of sets is empty if and only if each member is empty.

Now let $\Pr$ be any probability measure on $\Omega$.
Since the sets $\llbracket \bP \in I_k\rrbracket$ are pairwise disjoint,
finite additivity gives
\[
\Pr(A)
= \sum_{k=2}^{K-1} \Pr\bigl(\llbracket \bP \in I_k\rrbracket\bigr)
= \sum_{k=2}^{K-1} \Pr(\bP \in I_k).
\]
Therefore $\Pr(A) > 0$ if and only if at least one summand
$\Pr(\bP \in I_k)$ with $2\leq k\leq K-1$ is strictly positive.

\medskip
\noindent\emph{(iii) Quantum-like DPM.}
Assume~\eqref{eq:orth-direct-sum}. Since $\mathscr{H} = \mathscr{H}_{\mathrm{D}}\oplus \mathscr{H}_{\mathrm{C}}$, every $\psi\in \mathscr{H}$ admits a
unique decomposition
\[
\psi = \psi_D + \psi_C,
\qquad
\psi_D\in \mathscr{H}_{\mathrm{D}},\ \psi_C\in \mathscr{H}_{\mathrm{C}}.
\]
For $\psi\in S(\mathscr{H})$, we have $\psi\in \mathscr{H}_{\mathrm{D}}$ if and only if $\psi_C = 0$, and
$\psi\in \mathscr{H}_{\mathrm{C}}$ if and only if $\psi_D = 0$. Consequently,
\[
\psi \in A_S
\quad\Longleftrightarrow\quad
\psi \notin \mathscr{H}_{\mathrm{D}}\cap S(\mathscr{H})\ \text{and}\ \psi \notin \mathscr{H}_{\mathrm{C}}\cap S(\mathscr{H})
\quad\Longleftrightarrow\quad
\psi_D\neq 0\ \text{and}\ \psi_C\neq 0,
\]
proving~(iii)(a).

\medskip
\noindent\emph{Density of $A_S$.}
We show that $(\mathscr{H}_{\mathrm{D}}\cap S(\mathscr{H}))\cup(\mathscr{H}_{\mathrm{C}}\cap S(\mathscr{H}))$ has empty interior in $S(\mathscr{H})$ (with the
subspace topology), which implies that its complement $A_S$ is dense.

The set $\mathscr{H}_{\mathrm{D}}\cap S(\mathscr{H})$ is closed in $S(\mathscr{H})$ because $\mathscr{H}_{\mathrm{D}}$ is closed in $\mathscr{H}$, and similarly
$\mathscr{H}_{\mathrm{C}}\cap S(\mathscr{H})$ is closed. We claim $\mathscr{H}_{\mathrm{D}}\cap S(\mathscr{H})$ has empty interior. Fix
$\psi_0\in \mathscr{H}_{\mathrm{D}}\cap S(\mathscr{H})$ and let $U$ be any neighborhood of $\psi_0$ in $S(\mathscr{H})$.
Then there exists $\delta>0$ such that
\[
B_{S(\mathscr{H})}(\psi_0,\delta):=\{\psi\in S(\mathscr{H}):\|\psi-\psi_0\|<\delta\}\subseteq U.
\]
Since $\mathscr{H}_{\mathrm{C}}\neq\{0\}$, pick $u\in \mathscr{H}_{\mathrm{C}}$ with $\|u\|=1$. For $\varepsilon>0$ define
\[
\psi_\varepsilon := \frac{\psi_0 + \varepsilon u}{\|\psi_0 + \varepsilon u\|}\in S(\mathscr{H}).
\]
Because $\psi_0\perp u$, we have $\|\psi_0+\varepsilon u\|=\sqrt{1+\varepsilon^2}$ and thus
\[
\|\psi_\varepsilon-\psi_0\|
=\left\|\frac{\psi_0+\varepsilon u}{\sqrt{1+\varepsilon^2}}-\psi_0\right\|
=\left\|\Bigl(\frac{1}{\sqrt{1+\varepsilon^2}}-1\Bigr)\psi_0
+\frac{\varepsilon}{\sqrt{1+\varepsilon^2}}u\right\|
\;\xrightarrow[\varepsilon\to 0]{}\;0.
\]
Hence for sufficiently small $\varepsilon>0$ we have $\psi_\varepsilon\in U$.
Moreover, $\psi_\varepsilon\notin \mathscr{H}_{\mathrm{D}}$ because it has a nonzero component along $u\in \mathscr{H}_{\mathrm{C}}$,
and $\psi_\varepsilon\notin \mathscr{H}_{\mathrm{C}}$ because it has a nonzero component along $\psi_0\in \mathscr{H}_{\mathrm{D}}$.
Therefore $\psi_\varepsilon\in A_S\cap U$, showing that no neighborhood of $\psi_0$
is contained in $\mathscr{H}_{\mathrm{D}}\cap S(\mathscr{H})$. Thus $\mathscr{H}_{\mathrm{D}}\cap S(\mathscr{H})$ has empty interior, and by symmetry
the same holds for $\mathscr{H}_{\mathrm{C}}\cap S(\mathscr{H})$. Consequently their union has empty interior, and $A_S$ is dense in $S(\mathscr{H})$, proving~(iii)(b).

\medskip
\noindent\emph{Full measure of $A_S$.}
To prove~(iii)(c), it suffices to show $\mu(\mathscr{H}_{\mathrm{D}}\cap S(\mathscr{H}))=\mu(\mathscr{H}_{\mathrm{C}}\cap S(\mathscr{H}))=0$.
Let $X$ be an $\mathscr{H}$-valued random vector whose law is continuous and invariant under
unitary transformations; for instance take $X\sim\mathcal{N}_{\mathbb{C}}(0,I)$
(standard complex Gaussian). Then $\Pr(X=0)=0$, so $X/\|X\|$ is well-defined
almost surely, and the law of $X/\|X\|$ is precisely $\mu$.

Since $\mathscr{H}_{\mathrm{D}}$ is a proper linear subspace of $\mathscr{H}$, there exists a nonzero linear functional
$\ell\colon \mathscr{H}\to\mathbb{C}$ such that $\mathscr{H}_{\mathrm{D}}\subseteq\ker\ell$. The complex random variable
$\ell(X)$ is nondegenerate and has a continuous distribution on $\mathbb{C}$, hence
$\Pr(\ell(X)=0)=0$, and therefore
\[
\Pr(X\in \mathscr{H}_{\mathrm{D}})\le \Pr(\ell(X)=0)=0.
\]
Because $\mathscr{H}_{\mathrm{D}}$ is a cone, $X/\|X\|\in \mathscr{H}_{\mathrm{D}}$ if and only if $X\in \mathscr{H}_{\mathrm{D}}$, so
\[
\mu(\mathscr{H}_{\mathrm{D}}\cap S(\mathscr{H}))
=\Pr(X/\|X\|\in \mathscr{H}_{\mathrm{D}})
=\Pr(X\in \mathscr{H}_{\mathrm{D}})
=0.
\]
Similarly, $\mu(\mathscr{H}_{\mathrm{C}}\cap S(\mathscr{H}))=0$, whence by subadditivity,
\[
\mu\bigl((\mathscr{H}_{\mathrm{D}}\cap S(\mathscr{H}))\cup(\mathscr{H}_{\mathrm{C}}\cap S(\mathscr{H}))\bigr)=0,
\qquad
\mu(A_S)=1.
\]
This completes the proof.
\end{proof}

\section{Conclusions}\label{section:summary}

This paper has revisited the disjunction effect with a specific focus on the
Prisoner's Dilemma (PD) experiment. A common reading holds that the PD
findings are incompatible with the classical law of total probability (LTP)
\eqref{eq:LTP_DC}. We have shown that this conclusion rests on an implicit
certainty-only premise built into the conventional classical
decision-process model (DPM) under the mental-event reading of the
opponent-choice events. When the three information conditions are
identified with the corresponding conditional probabilities
\eqref{eq:identification}, the usual partition assumptions
\eqref{eq:empty_DC}-\eqref{eq:sure_DC} effectively require each participant
to be already certain that the opponent will defect or will cooperate before
making a choice. Once this strong premise is made explicit, the apparent violation of the LTP can be traced to that modeling choice rather than to
any failure of classical logic or probability itself.

To relax this restriction, we proposed a new classical DPM that explicitly
represents \emph{ambiguous} pre-decision expectations. Each participant is
characterized by a continuous expectation parameter
$\bP\in[0,1]$ encoding the anticipated likelihood of opponent defection, and
the participant pool is partitioned into bins $\{I_k\}$ of expectation
levels. The overall defection rate then obeys the classical LTP in the form
\eqref{eq:LTP_P}, expressed as a mixture over expectation bins. The
conventional DPM is recovered as the boundary case in which all probability
mass concentrates at the two extremes of $\bP$. Theorem~\ref{thm:disjunction_realization}
shows that this is not merely a conceptual possibility: for any triple of
empirical defection rates $(d_U,d_D,d_C)\in(0,1)^3$, including the strong
disjunction-effect case $d_U<\min\{d_D,d_C\}$, there exists a simple
three-bin instance of the proposed classical DPM that matches the triple
exactly while fully respecting the classical LTP.

We then placed this classical account alongside a standard quantum-like DPM
for the PD experiment, formulated in a common notation. The quantum-like
model represents opponent-choice events as subspaces of a Hilbert space and
yields a quantum analogue of the LTP \eqref{eq:quantum_LTP}, in which the same
three defection rates are related by an interference term. Interpreted as a
decision-process model, this framework also accommodates decisions under
ambiguous expectations, but locates that ambiguity within a single cognitive
state as a superposition over expectation subspaces and ties events to
measurement (questioning) outcomes. Theorem~\ref{cor:quantum_disjunction_realization}
then establishes that, at the level of the three observable defection rates,
the proposed classical DPM has no less representational power than the
quantum-like DPM: for any triple $(d_U,d_D,d_C)\in(0,1)^3$ produced by the
quantum-like DPM, there exists a three-bin classical DPM of the proposed
form that reproduces the same probabilities. 
In this sense, the disjunction-effect pattern in the PD experiment does not itself force a departure from classical probability; what
distinguishes the two approaches is not which triples can be realized,
but how events and ambiguity are represented.
Note that, to the best of our knowledge, 
any mathematical result like Theorem~\ref{thm:disjunction_realization} or 
Theorem~\ref{cor:quantum_disjunction_realization}
has not been known in the quantum cognition literature.

The present analysis has several implications and suggests directions for
future work:
\begin{itemize}
  \item \textbf{Reinterpreting the disjunction effect.}
  The PD disjunction effect need not be read as a breakdown of classical
  probability. Theorem~\ref{thm:disjunction_realization} demonstrates that
  even strong disjunction-effect patterns are compatible with the classical
  LTP once ambiguous expectations are admitted into the decision process.
  What fails is the combination of the certainty-only premise with the
  standard identification of information conditions, not Kolmogorov
  probability itself.

  \item \textbf{Clarifying the role of quantum-like models.}
  Theorem~\ref{cor:quantum_disjunction_realization} shows that, for the PD
  disjunction-effect experiment, quantum-like models do not gain empirical
  reach solely by abandoning the classical LTP: any triple of defection rates
  they generate can already be matched by a classical mixture over
  expectation levels. Their distinctive contribution lies instead in their
  event semantics and in the intra-state representation of ambiguity via
  superposition and interference, which may become crucial in richer
  experimental settings (e.g.\ sequential questioning, order effects, or joint
  tasks) where classical mixtures and quantum dynamics impose genuinely
  different constraints.

  \item \textbf{Classical and quantum-like approaches as complementary tools.}
  By formulating a parallel classical DPM and making explicit the underlying
  event semantics (see Section~\ref{sec:semantics_comparison}), the present framework provides a neutral
  baseline for evaluating when quantum-like structure is substantively needed
  in cognitive modeling. In some tasks, a refined classical model with
  heterogeneous expectations may suffice; in others, the quantum formalism
  may offer more natural or more constrained descriptions. Future work can
  exploit this side-by-side formulation to design experiments that probe the
  distribution of $\bP$, test specific quantum predictions against classical
  mixtures, and extend the analysis beyond PD to other disjunction-effect
  paradigms and related phenomena.
\end{itemize}

In summary, the PD disjunction-effect experiment is better viewed not as direct
evidence against classical probability, but as a challenge to how pre-decision
ambiguity is modeled. The choice between classical and quantum-like
formalisms should therefore be guided by explicit assumptions about events,
measurements, and uncertainty in the task at hand, rather than by a
presumed violation of the law of total probability.

\makeatletter
\@ifundefined{doi}
  {\newcommand{\doi}[1]{\url{https://doi.org/#1}}}
  {\renewcommand{\doi}[1]{\url{https://doi.org/#1}}}
\makeatother

\bibliographystyle{unsrtnat}
\bibliography{refs}

\end{document}